\documentclass[letterpaper, 10 pt, conference]{ieeeconf} 
 
\usepackage{subfigure}
\usepackage{algorithm}
\usepackage{algpseudocode}
\usepackage{amsmath,amssymb,amsthm}
\usepackage{graphicx}
\usepackage{subcaption}
\usepackage{booktabs}
\usepackage{placeins}
\usepackage{hyperref}
\usepackage{multirow}
\usepackage{caption}
\usepackage{todonotes}
\usepackage{pifont}
\usepackage{mathtools}

\newtheorem{assumption}{Assumption}
\newtheorem{definition}{Definition}

\newtheorem{theorem}{Theorem}
\newtheorem{lemma}{Lemma}

\makeatletter
\def\endfigure{\end@float}
\makeatother

\usepackage{ifthen}
\newboolean{showcomments}
\setboolean{showcomments}{false}

\usepackage{color}
\usepackage{todonotes}

\newboolean{arxiv}
\setboolean{arxiv}{false}

\definecolor{bleudefrance}{rgb}{0.19, 0.55, 0.91}
\definecolor{ao(english)}{rgb}{0.0, 0.5, 0.0}

\newcommand{\addcite}[0]{\ifthenelse{\boolean{showcomments}}
{\textcolor{purple}{(add cite(s)) }}{}}%

\newcommand{\addcites}[0]{\ifthenelse{\boolean{showcomments}}
{\textcolor{purple}{(add cite(s)) }}{}}%

\newcommand{\addref}[0]{\ifthenelse{\boolean{showcomments}}
{\textcolor{purple}{(add ref) }}{}}%

\newcommand{\enrique}[1]{  \ifthenelse{\boolean{showcomments}}
{\todo[inline,color=bleudefrance,caption={}]{Enrique: #1}}{}}
\newcommand{\rene}[1]{  \ifthenelse{\boolean{showcomments}}
{\todo[inline,color=cyan]{Ren\'e: #1}}{}}
\newcommand{\emmargin}[1]{\ifthenelse{\boolean{showcomments}}{\marginpar{\color{bleudefrance}\tiny EM: #1}}{}}
\newcommand{\hancheng}[1]{  \ifthenelse{\boolean{showcomments}}
{\todo[inline,color=orange]{Hancheng: #1}}{}}
\newcommand{\ziqing}[1]{  \ifthenelse{\boolean{showcomments}}
{\todo[inline,color=red]{Ziqing: #1}}{}}
\newcommand{\salma}[1]{  \ifthenelse{\boolean{showcomments}}
{\todo[inline,color=yellow]{Salma: #1}}{}}
\newcommand{\zxmargin}[1]{\ifthenelse{\boolean{showcomments}}{\marginpar{\color{purple}\tiny ZX: #1}}{}}
\newcommand{\stmargin}[1]{\ifthenelse{\boolean{showcomments}}{\marginpar{\color{red}\tiny ST: #1}}{}}
\newcommand{\jixian}[1]{  \ifthenelse{\boolean{showcomments}}
{\todo[inline,color=lightgray]{Jixian: #1}}{}}

\newcommand{\hl}[1]{\ifthenelse{\boolean{showcomments}}
{\textcolor{red}{#1}}{#1}}

\newboolean{showedits}
\setboolean{showedits}{false}
\usepackage[markup=underlined]{changes}
\definechangesauthor[color=bleudefrance]{EM}
\newcommand{\aem}[1]{
\ifthenelse{\boolean{showedits}}
{\added[id=EM]{#1}}
{\!#1\hspace{-4.75pt}}
}
\newcommand{\repem}[2]{
\ifthenelse{\boolean{showedits}}
{\replaced[id=EM]{#1}{#2}}
{\!#1\hspace{-4.75pt}}
}
\newcommand{\dem}[1]{
\ifthenelse{\boolean{showedits}}
{\deleted[id=EM]{#1}}
{}
}

\usepackage[backend=biber, style=ieee, doi=false, url=false, eprint=false]{biblatex}
\newcommand{\K}{\mathcal{K}}
\newcommand{\X}{\mathcal{X}}
\newcommand{\U}{\mathcal{U}}
\newcommand{\R}{\mathcal{R}}

\usepackage{amsmath}
\newcommand\restr[2]{{
  \left.\kern-\nulldelimiterspace 
  #1 
  \littletaller 
  \right|_{#2} 
  }}
\newcommand{\littletaller}{\mathchoice{\vphantom{\big|}}{}{}{}}

\let\labelindent\relax
\usepackage{enumitem}
\IEEEoverridecommandlockouts
\title{\LARGE \bf 
Recurrent Control Barrier Functions: A Path Towards Nonparametric Safety Verification
}

\author{Jixian Liu and Enrique Mallada
\thanks{J. Liu and E. Mallada are with the Department of Electrical and Computer Engineering, Johns Hopkins University, MD 21218, U.S.A. 
        {\tt\small jliu376@jh.edu, mallada@jhu.edu}.}
\thanks{This work was supported by NSF through grant Global Center 2330450, and Johns Hopkins University Institute for Assured Autonomy.
}}

\begin{document}
\maketitle
\thispagestyle{empty}
\pagestyle{empty}

\begin{abstract}
Ensuring the safety of complex dynamical systems often relies on Hamilton-Jacobi (HJ) Reachability Analysis or Control Barrier Functions (CBFs). Both methods require computing a function that characterizes a safe set that can be made (control) invariant. However, the computational burden of solving high-dimensional partial differential equations (for HJ Reachability) or large-scale semidefinite programs (for CBFs) makes finding such functions challenging. In this paper, we introduce the notion of Recurrent Control Barrier Functions (RCBFs), a novel class of CBFs that leverages a recurrent property of the trajectories, i.e., coming back to a safe set, for safety verification. Under mild assumptions, we show that the RCBF condition holds for the signed‑distance function, turning function design into set identification. Notably, the resulting set need not be invariant to certify safety. 
We further propose a data-driven nonparametric method to compute safe sets that is massively parellizable, and trades off conservativeness against computational cost. 
\end{abstract}
\enrique{One aspect not addressed in the abstract is the fact that there are also data-driven methods for safety. Now, it is clear that these methods usually lack guarantees, but this is not disucssed.}
\jixian{Copy that}
\enrique{Given that we are at the final stages of the paper. Please highlight in color significant changes to the paper so that I can just go over them in my final pass. }
\section{Introduction}
Safety is a fundamental requirement in the control of dynamical systems, particularly in safety-critical applications such as robotics, autonomous vehicles, etc. Safety of the system is typically enforced via Hamilton-Jacobi (HJ) reachability analysis~\cite{mbt2005tac} or Control Barrier Functions (CBFs)~\cite{acenst2019ecc}, both of which build a function whose superlevel set is a control invariant safe set. Unfortunately, despite the popularity of these methods, their application relies on the computation of the value function or CBF, which presents significant challenges. HJ-reachability analysis requires solving partial differential equations, which suffers from the curse of dimensionality~\cite{bcht2017cdc}. The synthesis of valid CBFs often requires solving a Sum-of-Squares (SOS) optimization problem, which is also computationally demanding when applied to high-dimensional systems~\cite{dp2023acc, c2021cdc}. 

To reduce the computational burden, some data-driven methods have been proposed. DeepReach greatly improves computational efficiency for high-dimensional HJ reachability by using neural PDE solvers, but its learning-based approximation limits interpretability despite strong empirical performance~\cite{bansal2021deepreach}.~\cite{folkestad2020data} accelerates the synthesis of CBF by utilizing Koopman-based matrix multiplications, though at the expense of losing strict guarantees due to operator approximation.~\cite{lee2024data} constructs control-invariant safe sets from hard constraints via data-driven CBFs, offering efficiency but with safety guarantees limited by uneven or sparse sampling quality. 

In this paper, we build a framework to trade off the computational complexity of finding safe control sets with the level of conservativeness of the solution, which has theoretical safety guarantees. A key insight of the proposed approach is to substitute the invariance property that Reachability and CBF methods aim to guarantee with a more flexible notion called recurrence~\cite{sspm2023cdc,ssm2024allerton}. A set is ($\tau$-) recurrent if every trajectory that leaves the set comes back to it (within $\tau$ units of time) infinitely many times. Recurrence has emerged as a practical surrogate for invariance in analysis and verification—e.g., for regions of attraction~\cite{sbm2022cdc}, stability~\cite{sspm2023cdc}, and safety verification~\cite{ssm2024allerton}. Information-theoretically, enforcing (control) recurrence demands lower data rates than invariance~\cite{sm2024hscc} and can often be achieved from finite trajectories~\cite{sm2025nahs}.

Building on this literature, we extend the notion of Recurrent Barrier Functions proposed in~\cite{ssm2024allerton} to account for the addition of controls, thus introducing Recurrent Control Barrier Functions (RCBFs). RCBFs relax strict invariance by requiring a finite-time ($\tau$) return to a safe set—conditions met by signed distance functions of given sets—while preserving safety as long as the set excludes the $\tau$-backward reachable tube of the unsafe region.
We devise a nonparametric, sampling-based procedure to synthesize RCBFs and verify safety quickly and at scale. To do so, we introduce a robust RCBF condition that uses trajectory data to certify a neighborhood of the initial state; an adaptive sampling method and data-driven exploration remove the need for large optimization programs. The method is GPU-friendly and lets practitioners trade conservativeness for computation without compromising safety.

The remainder of this paper is organized as follows. Section~\ref{sec:prelims} reviews preliminaries on HJ reachability analysis and CBFs. Section~\ref{sec:problemformulation} introduces the definition of RCBFs, extending classical CBFs through recurrence-based safety conditions. Section~\ref{sec:approach} develops the robust conditions that allow for data-driven verification of the RCBF property on a neighborhood of trajectory samples. Section~\ref{sec:algorithms} integrates the robust conditions into a sampling-based method for nonparametric safety verification that actively chooses where to sample based on prior outcomes. Section~\ref{sec:simulations} provides numerical validations demonstrating the effectiveness of the proposed approach. Section~\ref{sec:conclusion} concludes the paper and outlines directions for future work.

\paragraph*{Notation} 
$\|\cdot\|$ is an arbitrary norm on $\mathbb{R}^n$.  
For $x \in \mathbb{R}^n$ and $r > 0$, the closed ball of radius $r$ centered at $x$ is defined as $\mathcal{B}_r(x) := \{ y \in \mathbb{R}^n \mid \|y - x\| \leq r \}.$ Given a set $S \subseteq \mathbb{R}^n$ and a point $x \in \mathbb{R}^n$, the signed distance from $x$ to $S$ is  
\begin{align*}
\mathrm{sd}(x,S) :=
\begin{cases}
\inf_{y \in \partial S} \|y - x\|, & \text{if } x \notin S, \\
- \inf_{y \in \partial S} \|y - x\|, & \text{if } x \in S.
\end{cases}
\end{align*}

\section{Preliminaries and Related Work}
\label{sec:prelims}
\subsection{Problem Statement}
\label{sec:problem statement}
Consider a continuous-time control system:
\begin{align}
    \dot{x} = F(x,u),
    \label{eq:Control_System}
\end{align}
where $x \in \mathcal{X} \subseteq \mathbb{R}^{n}$ is the system's state in the state space $\mathcal{X}$, $u \in U \subseteq \mathbb{R}^m$ is the control input. We define $\mathcal{U}^{(a, b]} := \{u: (a,b] \rightarrow U| u \text{ is measurable}\}$, as the set of control signals on the time interval $(a,b]$, and $\mathcal{U} := \mathcal{U}^{(0,+\infty)}$. Given $u_0 \in \mathcal{U}^{(0,a]}$ and $u_1 \in \mathcal{U}^{(0, b]}$, their concatenation $u_0u_1 \in \mathcal{U}^{(0, a+b]}$ is defined as
\begin{align*}
    (u_0u_1)(t) = \begin{cases}
u_0(t), & t \in (0,a], \\
u_1(t), & t \in (a,a+b].
\end{cases}
\end{align*}
Similarly, for $u\in\U^{(a,b]}$ and $(c,d]\subset(a,b]$ we will use $\restr{u}{(c,d]}$ to denote the restriction of $u$ to the interval $(c,d]$.

In a more general setting, consider a sequence of control inputs $u_n \in \mathcal{U}^{(0,\tau_n]}$, where $\tau_n > 0$ for every $n \in \mathbb{N}$. We define $ u_{[n]} := u_0 u_1 \cdots u_n,$ and $  u_{[\infty]} := \lim_{n \to \infty} u_{[n]}$. At times, we adopt a slight abuse of notation by writing $u$ both for instantaneous inputs in $U$ and for signals in $\mathcal{U}^{(a,b]}$; the intended interpretation will always be clear from context.

Given an initial state $x \in \mathbb{R}^n$ and a control signal $u \in \mathcal{U}^{(0,a]}$, we denote by $\phi(t,x,u)$ the trajectory solving~\eqref{eq:Control_System} for all $t \in (0,a]$. Throughout, we impose the following regularity assumptions on~\eqref{eq:Control_System}.

\begin{assumption}[Forward Completeness]\label{ass:forward}
The control system~\eqref{eq:Control_System} is \textbf{forward complete}, 
that is, for any initial condition $x \in \mathbb{R}^n$ and any input $u \in \mathcal{U}$, 
the solution $\phi(\cdot,x,u)$ exists and is unique on $[0,\infty)$.
\end{assumption}

\begin{assumption}[Uniform Local Lipschitz Continuity]\label{ass:lipschitz}
The vector field $F(x,u)$ in~\eqref{eq:Control_System} is locally Lipschitz in $x$, uniformly with respect to $u$. More precisely, for every compact set $S \subseteq \mathbb{R}^n$, there exists a constant $L \geq 0$ such that
\begin{align*}
\| F(y,u) - F(x,u) \| \leq L \|y - x\|, 
\quad \forall x,y \in S, \; \forall u \in U.
\end{align*}
\end{assumption}

\subsection{Safety Assessment}
Our goal is finding input signals $u(\cdot) \in \mathcal{U}$ such that the solution $\phi(t,x,u)$ to~\eqref{eq:Control_System} can avoid an unsafe region $\mathcal{X}_u \subset \mathcal{X}$ for all time $t\geq 0$. To that end, we aim to design an algorithm that can quickly find a strict subset of $\mathcal{X} \backslash \mathcal{X}_u$ that achieves this goal. We will therefore say that a state $x$ is considered to be safe if one can find a control $u \in \mathcal U$ such that the state trajectory $\phi(t, x, u)$ does not visit the unsafe region for all future time.

\begin{definition}[Safe State]
\label{def:safety}
A state $x \in \mathcal{X}$ is said to be \textbf{safe} w.r.t. the system~\eqref{eq:Control_System} if there exists a control $u \in \mathcal{U}$ such that the trajectory $\phi(t, x, u)$ never visits the unsafe region $\mathcal{X}_{u}$, i.e., $ \exists u \in \mathcal{U}$, s.t. $\forall t \geq 0,$ $\phi(t, x, u) \notin \mathcal{X}_{u}$.
\end{definition}

A common approach to ensure safety according to Definition~\ref{def:safety} is to find some set $\mathcal{C}$ that does not intersect with $\mathcal{X}_u$ and has the additional property that trajectories that start in $\mathcal{C}$ can be kept in $\mathcal{C}$. That is, $\mathcal{C}$ is control invariant.

\begin{definition}[Control Invariant Set]
A set $\mathcal{C} \subseteq \mathcal{X}$ is \textbf{control invariant} w.r.t.~\eqref{eq:Control_System} if for every $x \in \mathcal{C}$, there exists a control $u \in \mathcal{U}$ such that $\phi(t,x,u) \in \mathcal{C}$ for all $t \ge 0$.
\end{definition}

\subsection{Reachability Analysis}
\label{sec:reachability}
As mentioned before, a widely adopted method to verify safety is the Hamilton–Jacobi (HJ) reachability analysis. In this framework, one aims to compute the collection of all initial states from which, no matter what control one chooses, the trajectory will eventually end in the unsafe set $\mathcal{X}_u$. We provide a formal definition next.

\begin{definition}[Backward Reachable Tube]
For a set $S$, and constant $T>0$, the $T$-Backward Reachable Tube ($T$-BRT) is defined as:
\begin{align*}
\mathcal{R}_{T}(S) := \{ x \mid \forall u \in \mathcal{U}^{(0, T]}, \exists t \in (0,T], \mathrm{ s.t. } \,\phi(t,x,u) \in S \}.
\end{align*}
When $T=\infty$, we refer to it simply as the Backward Reachable Tube (BRT) and denote by $\mathcal{R}(S)$.
\end{definition}

To construct the BRT, the HJ reachability procedure casts the safety verification task as an optimal control problem. Here, the controller’s objective is to avoid the unsafe set $\mathcal{X}_{u}$. This is quantitatively expressed through a value function $V(x,t) := \min_{ s \in [-t, 0]} l(\phi(s, x, u))$, which measures the minimum cost or the distance to entry $\mathcal{X}_{u}$. In the absence of disturbances, the evolution of $V(x,t)$ is governed by a Hamilton-Jacobi-Isaacs Variational Inequality that takes the form of a Hamilton–Jacobi–Bellman equation~\cite{bcht2017cdc}:
\begin{align*}
\min \{ D_t V(x,t) + H(x,t,\nabla V(x,t)),
l(x) - V(x,t) \} = 0,
\end{align*}
where $l(x)$ is the terminal condition where $V(x, 0) = l(x)$, and $H(x, t, \nabla V(x, t)) := \max_{u \in U} D_{x} V(x, t) \cdot f(x, u)$.

Once $V(x,T)$ is computed, the $T$-BRT is given by the sublevel set
\begin{align*}
\mathcal{R}_T(\mathcal{X}_u)=\{ x \mid V(x,T) \le 0, x \in \mathcal{X}_{u} \},
\end{align*}
which implies that any state within this set will eventually lead to $\mathcal{X}_{u}$ under any control $u(\cdot)$ within less than $T$ units of time. HJ reachability gives rigorous safety guarantees when $x \in \mathcal{R}^{c}_{+\infty}(\mathcal{X}_u)$, which is the largest safe control invariant set, but is computational costly in high dimensions. Efficient solvers for the HJ PDE mitigate this~\cite{bcht2017cdc}, improving practicality. 
Our work tackles safety from a complementary angle.

\subsection{Control Barrier Functions}\label{sec:CBF}
CBFs offer another conservative alternative to HJ reachability. By bounding $\dot h$ with an extended class-$\mathcal K$ function, they render a chosen set $\mathcal C$ control invariant and thus ensure safety. To formally introduce CBFs we are required to introduce the notion of extended class $\K$ functions. 

\begin{definition}[Extended Class $\K$ Function]
A function $\kappa: \mathbb{R} \to \mathbb{R}$ is an extended class $\K$ function if it is continuous, strictly increasing, and satisfies $\kappa(0) = 0$.
\end{definition}

We are now ready to formally introduce CBFs. 

\begin{definition}[Control Barrier Function~\cite{acenst2019ecc}]\label{def:CBF}
A continuously differentiable function $h(x)$ is a CBF for the system~\eqref{eq:Control_System} if there exists an extended class $\K$ function $\kappa$ such that,
\begin{align}\label{eq:CBF}
    \max_{u \in U} L_{F} h(x) +\kappa(h(x)) \ge 0,
\end{align}
for all $x \in \mathcal{X}$, and where $L_{F} h(x) = \frac{\partial h}{\partial x}^\top F(x,u),$ are first-order Lie derivatives.
\end{definition}

\begin{theorem}[~\cite{acenst2019ecc}]
An immediate consequence of Definition~\ref{def:CBF} is that any Lipschitz-continuous controller $k(x)$ satisfying $$k(x) \in \{ u\in U \mid L_{F}h(x) + \kappa(h(x)) \ge 0 \},$$ renders the set $h_{\ge 0}:=\{x:h(x)\ge0\}$ invariant. Thus, $h_{\ge0}$ is, by definition, control invariant.
\end{theorem}

Thus, if such a CBF $h$ exists and $h_{\ge 0} \cap \mathcal{X}_u = \emptyset$, all states in $h_{\ge 0}$ can find a control signal $u \in \U$, whose signal at any moment is in the set of $k(x)$. That is to say for any intial state $x \in h_{\ge 0}$, there exists $u \in \mathcal{U}$ such that $\phi(t,x,u) \in h_{\ge 0}, \forall t >0,$ which means the states in $h_{\ge 0}$ are safe~\cite{acenst2019ecc}.

Sum-of-Squares (SOS) programming is widely used to synthesize/verify polynomial CBFs, but its cost grows rapidly with system dimension~\cite{c2024tac,pj2004whscc}, and polynomials may poorly capture complex safety sets. Neural network CBFs improve expressivity~\cite{xwhcalr2023tro,lld2023corl,ssmgrrf2024icra}, yet their validity is harder to certify due to limited interpretability~\cite{ssmgrrf2024icra}.

\section{Recurrent Control Barrier Function}
\label{sec:problemformulation}

The core idea behind ensuring safety using traditional CBFs is to construct a scalar function that makes $h_{\geq0}$ control invariant. Such sets can be as computationally expensive as a BRT, making CBF synthesis difficult. Leveraging recurrence, we show this explicit invariant set is unnecessary: valid RCBFs can be built from \textbf{control recurrent sets}, which relax invariance while keeping safety guarantees.

\subsection{Control Recurrent Sets}

In this section, we briefly cover the definition of recurrent sets in a control systems setting, which broadly allow trajectories to leave a set, provided they come back to it. The presentation follows~\cite{sbm2022cdc,ssm2024allerton,sm2024hscc}, particularly \cite{sm2024hscc}.

\begin{definition}[Control Recurrent Sets]
\label{def:recurrence}
A compact set $S \subseteq \mathbb{R}^n$ is called \textbf{control recurrent} w.r.t.~\eqref{eq:Control_System} if, for all $ x \in S $, $\exists$ $u\in \mathcal{U}$, such that for any $ t \geq 0 $,
\begin{align}
\exists t' > t \;\; \mathrm{with} \;\; \phi(t', x, u) \in S.
\end{align}
Likewise, a set $ S \subseteq \mathbb{R}^n $ is called \textbf{control} $\tau$-\textbf{recurrent} ($\tau > 0$) w.r.t.~\eqref{eq:Control_System} if, for all $ x \in S $,  $\exists$ $u\in \mathcal{U}$, such that for any $t \ge 0$,
\begin{align}\label{eq:tau-recurrent}
\exists\, t' > t\,, \;\; \mathrm{with} \;\; t' - t \in (0, \tau]\,, \;\; \mathrm{and} \;\; \phi(t', x, u) \in S\,.
\end{align}
We refer to such $\phi(t,x,u)$ as a ($\tau$-)recurrent trajectory.
\end{definition}

As Figure~\ref{fig:recurrence} shows, although a $\tau$-recurrent set is not necessarily invariant, it ensures that trajectories starting in $S$ will revisit it within at most $\tau$-time units infinite times. Notably, based on the Definition~\ref{def:recurrence}, an invariant set is always $\tau$-recurrent for any $\tau>0$. Additionally, a 0-recurrent set is equivalent to an invariant set. Thus, Definition~\ref{def:recurrence} generalizes invariance by allowing the trajectory $ \phi(t, x, u) $ to leave the set $ S $ before returning~\cite{ssm2024allerton}.
\label{sec:recurrence}
\begin{figure}[tbp]
    \centering
\includegraphics[width=0.9\linewidth]{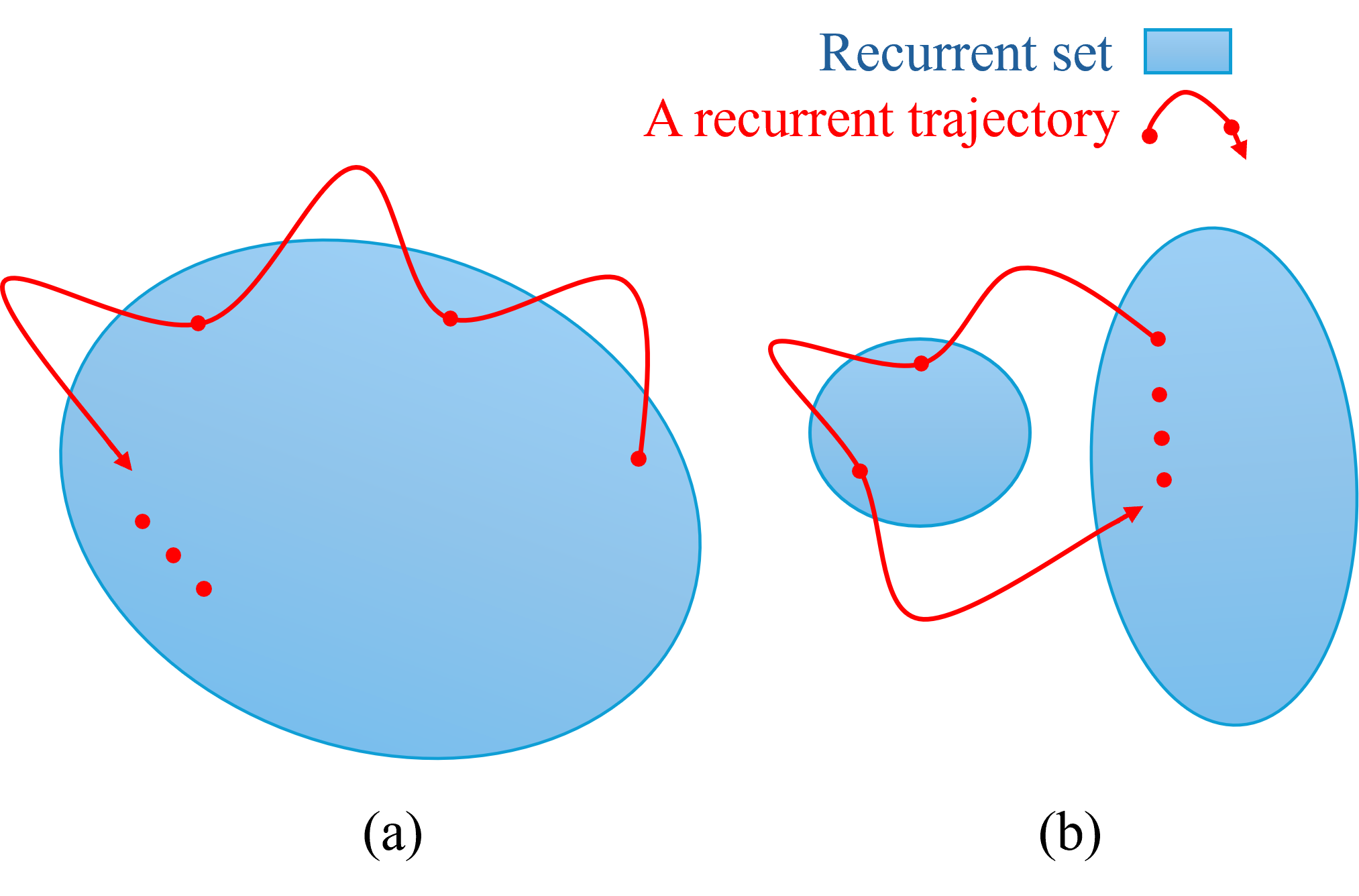}
    \caption{Illustration of Recurrent Sets and Recurrent Trajectories}
\label{fig:recurrence}
\vspace{-1em}
\end{figure}
Compared with the invariant sets, recurrent sets show a more flexible shape; it does not need the region to be connected, and it does not require the system~\eqref{eq:Control_System} to point inwards (or at least not outwards) on all the boundary $\partial S$. 



\subsection{Recurrent Control Barrier Function}
We now move towards introducing the proposed Recurrent Control Barrier Functions. In fact, similar to~\cite{ssm2024allerton}, simply requiring trajectories to return to the set within a finite time, infinitely many times can guarantee the safety for the dynamical system.




\begin{definition}[Recurrent Control Barrier Function]
\label{def:RCBF}
Consider the control system~\eqref{eq:Control_System}. A continuous function $h: \mathbb{R}^{n} \rightarrow \mathbb{R}$ is a \textbf{Recurrent Control Barrier Function (RCBF)} if for all $x\in D_0 := h_{\ge-c}$, with $c>0$, $\exists\, u \in \U^{(0,\tau]}$ s.t.
\begin{align}
\label{eq:RCBF}
\max\limits_{t\in (0,\tau]} e^{\gamma(h(\phi(t,x,u)))t} \,h(\phi(t,x,u)) \ge h(x),
\end{align}
where the function $\gamma: \mathbb{R} \to \mathbb{R}_{>0}$.
\end{definition}
In \eqref{eq:RCBF} we follow the standard convention that when the $\sup$ is not achieved within the set $(0,\tau]$ the $\max$ is $-\infty$. Thus, for the $\max$ to be lower bounded, it implies that it is achieved within $(0,\tau]$.
A particular choice of $\gamma$ that will be of use throughout this paper is
\begin{equation}\label{eq:gamma-piecewise}
\gamma_{\alpha,\beta}(s) = \begin{cases}
        \alpha, \; \text{ if }s\geq0\,,\quad\text {and}\quad
        \beta, \; \text{ if }s<0\,,
    \end{cases}
\end{equation}
where $\alpha$ and $\beta$ are positive parameters. This will be particularly useful in our converse results in Section \ref{Sec: Validity}.

The following theorem describes how to use RCBFs to assess safety.
\begin{theorem}[Safety Assessment via RCBFs]
\label{theo:safety_assessment}
Let $h$ be an RCBF as in Definition~\ref{def:RCBF}. Then:
\begin{enumerate}[label=(\roman*)]
    \item The superlevel set $h_{\ge 0}$ is control $\tau$-recurrent, i.e., for any $x \in h_{\ge 0}$ there exists $u \in \U$ such that the trajectory $\phi(t,x,u)$ always returns to $h_{\ge 0}$ within time $\tau$.
\end{enumerate}
Moreover, if $h_{\ge 0} \cap \mathcal{R}_\tau(\mathcal{X}_u)=\emptyset$, then:
\begin{enumerate}[label=(\roman*)]
 \setcounter{enumi}{1}
    \item For any $x \in h_{\ge 0}$, every $u \in \U$ that renders $\phi(t,x,u)$ $\tau$-recurrent also ensures that $\phi(t,x,u)\notin \X_u$ for all $t\geq0$. 
\end{enumerate}
In particular, under the condition $h_{\ge 0} \cap \mathcal{R}_\tau(\mathcal{X}_u)=\emptyset$, the set $h_{\ge 0}$ is safe.
\end{theorem}
\begin{proof}
\textbf{(i) Control $\tau$-recurrence of $h_{\ge 0}$.}
Given $x\in h_{\ge 0}$, fix $x_0:=x$ and $t_0:=0$. By the RCBF condition~\eqref{eq:RCBF}, there exists
$u_0\in\U^{(0,\tau]}$ and a time
\begin{equation}\label{eq:t_0}
\tau_0 := \max\left\{ \arg\max_{t\in(0,\tau]}\ 
\ e^{\gamma(h(\phi(t,x_0,u)))\,t}\,h(\phi(t,x_0,u))\right\},
\end{equation}
such that $x_1:=\phi(\tau_0,x_0,u_0)\in h_{\ge 0}$ and $t_1:=\tau_0+t_0$. Proceed inductively: given $x_n\in h_{\ge0}$ and $t_n$, use \eqref{eq:RCBF} to select $u_n\in\U^{(0,\tau]}$ and
$\tau_n\in(0,\tau]$ as in \eqref{eq:t_0}, leading to $x_{n+1}:=\phi(\tau_{n},x_{n},u_{n})\in h_{\ge 0}$, and $t_{n+1}=\tau_n+t_n$.

The desired control $u\in \U$ is thus defined by concatenating the restrictions of ${u_n}$ to the intervals ${(0,\tau_n]}$, i.e.,
\[
u_{[n]}=\restr{u_0}{(0,\tau_0]}\restr{u_1}{(0,\tau_1]}\dots\restr{u_n}{(0,\tau_n]} \in \U^{(0,t_n]}
\]
and letting $u=\lim_{n\to\infty}u_{[n]}\in \U^{(0,t^*]}$, where $t^*=\lim_{n\rightarrow\infty} t_n$. An argument similar to~\cite[Lemma 1]{sspm2023cdc} shows that $t^*=\infty$. Moreover, it follows from the construction that for all $n\ge0$,
\begin{align}
    x_{n+1} =\phi(\tau_n,x_n,\restr{u_n}{(0,\tau_n]}) = \phi(t_n,x,u),
\end{align}
and therefore $\phi(t_n,x,u)\in h_{\ge0}$.
It follows then from the fact that for all $n\geq0$, $x_n\in h_{\ge0}$, $t_{n+1}-t_n\in(0,\tau]$ and $t_n\to \infty$, that the trajectory $\phi(t,x,u)$ is $\tau$-recurrent w.r.t. $h_{\ge0}$. Since $x\in h_{\ge0}$ was chosen arbitrarily, $(i)$ follows.

\medskip
\noindent\textbf{(ii) Safety under $h_{\ge 0}\cap \mathcal{R}_\tau(\mathcal{X}_u)=\emptyset$.}
Assume $h_{\ge 0}\cap \mathcal{R}_\tau(\mathcal{X}_u)=\emptyset$. Take any
$x\in h_{\ge 0}$ and any control $u\in\U$ that renders $\phi(t,x,u)$
$\tau$-recurrent w.r.t. $h_{\ge0}$. Suppose, towards a contradiction, that the trajectory is unsafe: there exists $t'>0$ with $\phi(t',x,u)\in\X_u$. Since $h(x)\ge 0$ and $h<0$ on $\X_u$, by continuity there exists a \emph{last exit time} $t''\in[0,t']$ with
$h(\phi(t'',x,u))=0$ and $h(\phi(t,x,u))<0$ for all $t\in(t'',t']$.
Because $h_{\ge 0}\cap \mathcal{R}_\tau(\mathcal{X}_u)=\emptyset$, the state
$\phi(t'',x,u)$ cannot reach $\X_u$ within time $\tau$, hence $t'-t''>\tau$ and
\[
h\big(\phi(t''+t,x,u)\big)<0\qquad\forall\,t\in(0,\tau],
\]
which contradicts with the fact that $u$ renders $\phi(t,x,u)$ $\tau$-recurrent w.r.t. $h_{\ge0}$.
\end{proof}

Besides ensuring safety, RCBFs also share similar properties like standard CBFs. In particular, it is possible to show that whenever $x\in D_0\backslash h_{\ge0}$, there is always some $u\in\U$ such that 
$$\lim\inf_{t\rightarrow\infty} \mathrm{d}(h_{\geq0},\phi(t,x,u))=0,$$
with $d(S,x):=\min_{y\in S}\|y-x\|$, thus ensuring that trajectories come back to $h_{\geq0}$ under simplified condition. We do not make these claims formal here, and refer the reader to \cite{ssm2024allerton} for similar arguments.

\subsection{Signed Distance Function: a Valid RCBF}
\label{Sec: Validity}
In this section, we present a striking result. The existence of a CBF $h$ satisfying some regularity conditions is sufficient for synthesizing a simple sign distance function that satisfies our RCBF condition. Firstly, we require $h$ to be sector contained.

\begin{definition}[Sector Containment] Let $h: D\subseteq \mathbb{R}^n \rightarrow \mathbb{R}$ be continuous. If $\exists a_1, a_2 > 0$ such that
\begin{align}\label{eq:sector_containment}
    (h(x) - a_{1}\mathrm{sd}(x, h_{\leq 0}))(h(x) - a_{2}\mathrm{sd}(x, h_{\leq 0})) \leq 0
\end{align}
for all $x \in D$, we say that $h$ is sector contained.
\label{def:sector_containment}
\end{definition}

The second condition refers to the particular choice of extended class $\K$ function. In particular, we will consider the sub-class:
\begin{equation}\label{eq:kappa-piecewise}
    \kappa_{\alpha,\beta}(s):= \gamma_{\alpha,\beta}(s)\, s.
\end{equation}

\begin{theorem}[Validity of Signed Distance Function as RCBF]
Let $h$ be a CBF satisfying~\eqref{eq:CBF} and \eqref{eq:sector_containment} over $D_0:=h_{\geq -c}$ with parameters $c>0$ and $a_2>a_1>0$, and extended class $\K$ function $\kappa_{\alpha,\beta}$ as in \eqref{eq:kappa-piecewise}, with parameters $\alpha>0$ and $\beta>0$.
Then, for any closed set $S$ satisfying ${h}_{\geq 0} \subseteq S \subseteq  h_{\geq -c}$, with $\partial S\cap h_{=0}=\emptyset$, the function $$\hat{h}(\cdot) = - \mathrm{sd}(\cdot, S)$$ is an RCBF over $\hat D_0:=\hat h_{\geq -\hat c}$ where $\hat c\geq 0$ is the largest constant satisfying $\hat h_{\geq -\hat c} \subseteq h_{\geq -c}$.

Precisely, for all $x\in \hat h_{\geq-\hat c}$, there exists $u\in\U^{(0,\tau]}$ s.t.
\begin{align}
    \max\limits_{ t\in (0,\hat\tau]} 
    e^{\hat \gamma( \hat h(\phi(t,x,u)))t} \hat h(\phi(t,x,u))
    \geq \hat{h}(x)
\label{eq:validity}
\end{align}
 where $\hat\gamma:=\gamma_{\hat\alpha,\hat\beta}$, with $\hat{\alpha}, \hat{\beta} > 0$ satisfying $\hat{\alpha} > \alpha, \hat{\beta} < \beta$, and $
    \hat{\tau} \geq \max\{\frac{\mathrm{log}(a_2/a_1)}{\hat{\alpha} - \alpha}, \frac{\mathrm{log}(a_2/a_1)}{\beta - \hat{\beta}}\} + \frac{\log (\overline{\delta}/\underline{\delta})}{\min \{\hat{\alpha}, \hat{\beta}\}}$
\label{theo:Validity}
where
\begin{align*}
\overline{\delta} &:= \sup_{x \in D_0} \left(\mathrm{sd}(x,S) - \mathrm{sd}(x,h_{\geq 0})\right),\\[6pt]
\underline{\delta} &:= \inf_{x \in D_0} \left(\mathrm{sd}(x,S) - \mathrm{sd}(x,h_{\geq 0})\right).
\end{align*}
\end{theorem}

\begin{proof}
The proof follows closely similar results for the non-control case~\cite[Theorem 11]{ssm2024allerton} and it is omitted due to space constraints.
\end{proof}

\section{Safety Enforcement Using Recurrence}
\label{sec:approach}
 
In this section, we aim to develop robust conditions that leverage trajectory samples to certify the satisfaction of the RCBF condition on a neighborhood of the trajectory. The proposed approach reduced the problem of checking the RCBF condition on uncountably many points, to checking it on finitely many states, possibly in parallel.

\subsection{Verification of a Cell}
To proceed, we first analyze how trajectories deviate from one another. This step lays the foundation for constructing a stronger verification criterion that ensures local safety in a neighborhood $\mathcal{B}_{r}(x)$ of each sampled point $x$, which we refer here as a cell.

\begin{lemma} 
Suppose that two trajectories $\phi(t, x, u)$, $ \phi(t,y, u)$ starts from $x$ and $y$ respectively and share the same control input $u$ all the times, where $\|x - y\| \leq r$. Then we have:
\begin{align}
    |\mathrm{sd}(\phi(t, y, u), S) \!-\! \mathrm{sd}(\phi(t, x, u), S) | \leq re^{Lt}, \forall t\geq0,
\end{align}
where $L$ is a uniform bound on the Lipschitz constant of~\eqref{eq:Control_System} on the $x$ variable.
\label{lemma:triangle_inequality}
\end{lemma}

\ifthenelse{\boolean{arxiv}}
{
\begin{proof}
The proof is omitted due to space constraints, and will be available in the archival version of this paper.
\end{proof}
}{
\begin{proof} See Appendix~\ref{lemma:tran_ineq}.
\end{proof}
}
%

We leverage Lemma~\ref{lemma:triangle_inequality} to verify different properties of a given cell $\mathcal{B}_r(\cdot)$.
In particular, Theorem~\ref{theo:reachability_Approximation} below gives us a condition that verifies whether a cell $\mathcal{B}_{r}(x)$ is completely inside $\mathcal{R}_{\tau}(\mathcal{X}_{u})$, completely outside $\mathcal{R}_{\tau}(\mathcal{X}_{u})$, or partially inside $\mathcal{R}_{\tau}(\mathcal{X}_{u})$. This will be critical to over approximate $\R_\tau(\X_u)$.
\begin{theorem}
Consider a state $x \in \mathcal{X}$, and let $\mathcal{B}_{r}(x) :=\{y | \|y-x\| \leq r\}$ be a neighborhood of $x$. Then, we have:
\begin{itemize}
    \item[(i)] Given $x\in \X$, if there exists $u\in \U^{(0,\tau]}$ s.t.
          \begin{align}
            \forall t \in [0,\tau],\;\mathrm{sd}(\phi(t, x, u), \mathcal{X}_{u}) > re^{Lt},  
            \label{eq:robust_safe}
          \end{align}
          for some $r>0$, then $\mathcal{B}_{r}(x) \cap \mathcal{R}_{\tau}(\mathcal{X}_{u}) = \emptyset $,  
    \item[(ii)] Conversely, given $x\in\X$, if for all $u \in \mathcal{U}^{(0,\tau]}$,
          \begin{align}
               \exists t\in (0,\tau], \text{ s.t. } \mathrm{sd}(\phi(t, x, u), \mathcal{X}_{u}) < -r e^{Lt}, 
             \label{eq:robust_unsafe}
          \end{align}
          for some $r > 0$, then $\mathcal{B}_{r}(x) \subseteq \mathcal{R}_{\tau}(\mathcal{X}_{u})$
\end{itemize}
\label{theo:reachability_Approximation}
\end{theorem}
\ifthenelse{\boolean{arxiv}}
{
\begin{proof}
The proof is omitted due to space constraints, and will be available in the archival version of this paper.
\end{proof}
}{
\begin{proof}
\noindent
    \textit{(i)} If the initial states satisfy the condition~\eqref{eq:robust_safe}, then by Lemma~\ref{lemma:triangle_inequality}, we have:
    \begin{align*}
      \mathrm{sd}(\phi(t, y, u), \mathcal{X}_{u})
      \geq \mathrm{sd}(\phi(t, x, u), \mathcal{X}_{u})-re^{-Lt}>0,
    \end{align*}
    for all $t \in [0,\tau]$ and all $y \in \mathcal{B}_{r}(x)$.  Hence,
    \begin{align*}
      \mathcal{B}_{r}(x) \cap \mathcal{R}_{\tau}(\mathcal{X}_{u}) = \emptyset.
    \end{align*}

    \noindent
    \textit{(ii)} If instead the initial states satisfy the condition~\eqref{eq:robust_unsafe}, let $t^{*} < \tau$ be the time at which
    \begin{align*}
      \mathrm{sd}(\phi(t^{*}, x, u), \mathcal{X}_{u}) <-re^{Lt^{*}}.
    \end{align*}
    Again, by Lemma~\ref{lemma:triangle_inequality}, we have:
    \begin{align*}
      \mathrm{sd}(\phi(t^{*}, y, u), \mathcal{X}_{u})
      \le \mathrm{sd}(\phi(t^{*}, x, u), \mathcal{X}_{u})+re^{Lt^{*}}
      <0,
    \end{align*}
    for all $y \in \mathcal{B}_{r}(x)$.  Consequently,
    \begin{align*}
      \mathcal{B}_{r}(x) \subseteq R_{t^{*}}(\mathcal{X}_{u}).
    \end{align*}
\end{proof}}

The following theorem verifies whether the states of a cell all satisfy the RCBF condition~\eqref{eq:RCBF} or all such states are guaranteed not to satisfy such a condition.

\jixian{-sd and sd???}

\begin{theorem}
Given a closed set $S$, a candidate RCBF $h(\cdot) := - \mathrm{sd}(\cdot, S)$, and function $\gamma:=\gamma_{\alpha,\beta}$, with $\alpha,\beta>0$.
\begin{itemize}
    \item[(i)] Let 
        \begin{align}
            \hat{h}_{r}^{-}(x,u,t) := h(\phi(t, x, u)) - re^{Lt}           
        \end{align}
    and assume that $\exists u\in\U^{(0,\tau]}$ s.t. the following holds
    \begin{align}
        \max\limits_{ t \in (0,\tau]} e^{\gamma(\hat h_r^-(x,u,t)) t}\, \hat{h}_{r}^{-}(x, u, t)\geq h(x) + r,
        \label{eq:robu_recur}
    \end{align}
    for some $r>0$. Then for all $y \in \mathcal{B}_{r}(x)$, the RCBF condition is satisfied, i.e.,
    \begin{align}
    \hspace{-1em}
        \max\limits_{t \in (0,\tau]} 
        e^{\gamma(h(y, u, t)) t} h(\phi(t, y, u))
        \geq h(y), 
        \label{eq:robu_recur_con}
    \end{align}
    \item[(ii)] Let 
    \begin{align}
         \hat{h}_{r}^{+}(x,u,t) := h(\phi(t, x, u)) + re^{Lt}
    \end{align}
    and assume that $\forall u\in\U^{(0,\tau]}$ s.t. the following holds
    \begin{align}
    \hspace{-1em}
        \max\limits_{ t \in (0,\tau]} e^{\gamma(\hat h_r^+(x,u,t)) t}\, \hat{h}_{r}^{+}(x, u, t)< h(x) - r
        \label{eq:robu_unrecur}
    \end{align}
    for some $r > 0$. Then, for all $y \in \mathcal{B}_{r}(x)$, the RCBF condition is not satisfied, i.e.,  
    \begin{align}
    \hspace{-1em}
        \max_{t \in (0, \tau]} 
        e^{\gamma(h(\phi(t,y,u))) t} h(\phi(t, y, u))
        < h(y).
        \label{eq:robu_unrecur_con}
    \end{align}
\end{itemize}
\label{theo:RCBF_Approximation}
\end{theorem}
\ifthenelse{\boolean{arxiv}}
{
\begin{proof}
The proof is omitted due to space constraints, and will be available in the archival version of this paper.
\end{proof}
}{
\begin{proof}
\begin{itemize}
    \item [(i)]Let $t^{*}$ and $u^{*}$ be the time that maximizes the left-hand side of~\eqref{eq:robu_recur}, i.e. 
\begin{align*}
  t^{*}
  =
  \arg\max_{t\in(0,\tau]} \max\limits_{u \in \mathcal{U}^{(0,\tau]}} e^{\gamma(\hat h_r^-(x,u,t)) t}\, \hat{h}_{r}^{-}(x, u, t),\\
    u^{*}
  =
  \arg\max_{u\in \mathcal{U}^{(0,\tau]}} \max\limits_{t \in (0,\tau]} e^{\gamma(\hat h_r^-(x,u,t)) t}\, \hat{h}_{r}^{-}(x, u, t).
\end{align*}
At this maximizing time $t^{*} \in (0,\tau]$, it follows that
\begin{align*}
         & \max\limits_{u \in \mathcal{U}, t \in (0,\tau]} 
        e^{\gamma(h(y, u, t)) t} h(\phi(t, y, u))\\
        \geq&e^{\gamma(h(y,u^{*},t^{*})) t^{*}}\, h(y, u^{*}, t^{*})
  \\\geq &
  e^{\gamma(\hat h_r^-(x,u,t^{*})) t^{*}}\, \hat h_{r}^{-}(x, u, t^{*})
  \\
  \geq &
  h(x)+r
  \\
  \geq &
  h(y),
\end{align*}
where the first inequality follows from the definition of maximum, the second and fourth inequalities are derived from Lemma~\ref{lemma:triangle_inequality} and the third inequality is derived from the conditions of~\eqref{eq:robu_recur_con}.
    \item [(ii)] Let $t^{*}$ and $u^{*}$ be the time and control signal that maximize the left-hand side of~\eqref{eq:robu_unrecur_con}, i.e. 
\begin{align*}
  t^{*} = \arg\max_{t\in(0,\tau]}\max_{u \in \mathcal{U}^{(0,\tau]}} e^{\gamma(h(y,u,t)) t}\,h(y, u, t),\\
  u^{*} = \arg\max_{u\in\mathcal{U}^{(0,\tau]}}\max_{t \in (0,\tau]}e^{\gamma(h(y,u,t)) t}\,h(y, u, t).
\end{align*}
Again, at this maximizing time $t^{*} \in (0,\tau]$, we get

\begin{align*}
  & \max_{u \in \mathcal{U}^{(0,\tau]}, t \in (0, \tau]} e^{\gamma(h(y, u, t)) t} h(y, u, t)\\
  \le &
e^{\gamma(\hat h_r^+(x, u^*, t^*)) t^*} \hat h_r^+(x, u^*, t^*)
  \\
  \leq & \max\limits_{u \in \mathcal{U}^{(0,\tau]}, t \in (0,\tau]} e^{\gamma(\hat h_r^+(x,u,t)) t}\, \hat{h}_{r}^{+}(x, u, t)\\
  < & h(x)-r\\
  < & h(y),
\end{align*}
where the second inequality is derived from the definition of the maximum, the first and fourth inequalities are derived from Lemma~\ref{lemma:triangle_inequality}, and the third inequality is derived from the conditions of~\eqref{eq:robu_unrecur}.
\end{itemize}
\end{proof}}

\section{Numerical Methods}
\label{sec:algorithms}

Building on Theorem~\ref{theo:reachability_Approximation} and Theorem~\ref{theo:RCBF_Approximation}, we propose a safety verification algorithm aimed at finding a set $S\subset \X$ such that $h=-\mathrm{sd}(x,S)$ satisfies all the necessary properties for safety assessment described in Theorem~\ref{theo:safety_assessment}.
The proposed method adaptively partitions $\mathcal{X}$ into cells
$
\mathcal{G}:=\{g_i:=\mathcal{B}_{r_i}(x_i)\}_{i=1}^{|\mathcal{G}|},$  such that $g_i\cap g_j=\emptyset, i \neq j$ and $\cup\mathcal{G}:=\cup_{i=1}^{|\mathcal{G}|} g_i=\mathcal{X},$ 
when running our algorithms, two disjoint lists, $\mathcal{G}_s$ (tentative safe cells) and $\mathcal{G}_u$ (verified unsafe cells) are maintained and refined, and progressively, cells from $\mathcal{G}_s$ are assigned to $\mathcal{G}_u$ (while keeping $\X = (\cup \mathcal{G}_u) \cup (\cup\mathcal{G}_s)$) until one is able to guarantee that the safe set $S=\cup \mathcal{G}_s$ and RCBF $h=-\mathrm{sd}(x,S)$ satisfy the robust conditions of Theorem~\ref{theo:safety_assessment}.

This verification process is carried out in three stages, as illustrated in Algorithm~\ref{alg:verifyregion}, where lines $2, 3,$ and $4$ represent stages $1, 2,$ and $3$, respectively.

\begin{algorithm}[H]
\caption{VerifyRegion($\mathcal{X}, \tau$, $\alpha$, $\beta$)}
\label{alg:verifyregion}
\begin{algorithmic}[1]
\State \textbf{Input:} State Space $\mathcal{X}$, Parameters $\tau$, $\alpha$, and $\beta > 0$. 
\State $\mathcal{G}_{s}, \mathcal{G}_{u} = \mathrm{VerifyCells}(\mathcal{X}, \mathcal{X}_{u}, 0, ~\eqref{eq:robust_safe}, ~\eqref{eq:robust_unsafe})$
\State $\mathcal{G}_{s}, \mathcal{G}_{u} = \mathrm{VerifyCells}(\mathcal{G}_{s}, \mathcal{G}_{u}, \tau, ~\eqref{eq:robust_safe}, ~\eqref{eq:robust_unsafe})$
\State $\mathcal{G}_{s}, \mathcal{G}_{u} = \mathrm{VerifyCells}(\mathcal{G}_{s}, \mathcal{G}_{u}, \tau, ~\eqref{eq:robu_recur_con}, ~\eqref{eq:robu_unrecur_con})$
\Comment{$\alpha$ and $\beta$ are used in the conditions \eqref{eq:robu_recur_con} and \eqref{eq:robu_unrecur_con}.}
\State \Return $\mathcal{G}_{s}$, $\mathcal{G}_{u}$
\end{algorithmic}
\end{algorithm}
Each stage aims to sequentially get a better approximation of a region $S \subseteq \mathcal{X}$ for $h=-\mathrm{sd}(x,S)$ to be a valid RCBF.
Stage 1 first finds a sufficiently fine outer approximation of $\X_u$. Stage 2 finds an outer approximation of $\R_\tau(\cup \mathcal{G}_u)$, with $\mathcal{G}_u$ being the output of Stage 1. Finally, Stage 3 further uses $S = \cup \mathcal{G}_s$ in order to find such a $h$ satisfy the RCBF condition. 

All stages are implemented by calling a VerifyCells routine, Algorithm~\ref{alg:verifycells}, with the current estimates of $\mathcal{G}_s$ and $\mathcal{G}_u$ and the assignment conditions $\mathcal{C}_s$ and $\mathcal{C}_u$, corresponding to conditions~\eqref{eq:robust_safe} and~\eqref{eq:robu_recur_con}, and~\eqref{eq:robust_unsafe} and~\ref{eq:robu_unrecur_con}, respectively. Note that we initially start with one cell (the full set $\X$), and each pass progressively finds finer and more accurate approximations for $\mathcal{G}_s$ and $\mathcal{G}_u$.

\begin{algorithm}[htbp]
\caption{VerifyCells($\mathcal{G}, \mathcal{G}_{u}, \tau, \mathcal{C}_{\mathrm{s}},\mathcal{C}_{\mathrm{u}}$)}
\label{alg:verifycells}
\begin{algorithmic}[1]
\State \textbf{Input:} Grid $\mathcal{G}$ and $\mathcal{G}_{u}$, Parameter $\tau \ge 0$, Robust Safe Condition $\mathcal{C}_{\mathrm{s}}$, and Robust Unsafe Condition $\mathcal{C}_{\mathrm{u}}$. 
\While{$\mathcal{G} \neq \emptyset$}
\For{$\forall g_{i} = \mathcal{B}_{r}
(x) \in \mathcal{G}, i \in \mathbb{N}$ for some $r$ and $x$}
\State $\mathcal{G} \gets \mathcal{G} - \{g_{i}\}$

\State \hspace{-0.9em}$\mathcal{G}$, $\mathcal{G}_s$, $\mathcal{G}_{u}$ = $\mathrm{SafetyCheck}(g_i, \mathcal{G}, \mathcal{G}_u, \tau, \mathcal{C}_{s}, \mathcal{C}_{u})$
\EndFor
\EndWhile
\State \Return $\mathcal{G}_{s}$, $\mathcal{G}_{u}$
\end{algorithmic}
\end{algorithm}

The verification process in Algorithm~\ref{alg:verifycells} can be done in parallel for all cells in the input $\mathcal{G}$ and ends when this set is empty. This framework facilitates high parallelism through concurrent cell verification while ensuring rigorous safety guarantees. That is to say, each cell in Algorithm~\ref{alg:verifycells} is eventually verified to be safe or declared to be unsafe by employing the safe and unsafe assignment conditions, $\mathcal{C}_{s}$ and $\mathcal{C}_{u}$, corresponding to each stage. The specific verification of each cell is implemented by routines $\mathrm{SafetyCheck}$ (Algorithm~\ref{alg:unifiedsafetycheck}). 

\begin{algorithm}[H]
\caption{SafetyCheck($g_i, \mathcal{G}, \mathcal{G}_{u}, \tau, \mathcal{C}_{s}, \mathcal{C}_{u}$)}
\label{alg:unifiedsafetycheck}
\begin{algorithmic}[1]
\State \textbf{Input:} Cell $g_i$ to check; Grids $\mathcal{G}, \mathcal{G}_{s}, \mathcal{G}_{u}$; Parameter $\tau > 0$; Robust Safe Condition $\mathcal{C}_{s}$; Robust Unsafe Condition $\mathcal{C}_{u}$.
\State $\mathcal{X}_{u} = \mathcal{G}_{u}$
\State $\mathcal{G}_{s} = \emptyset$
\State Sample $n_s$ trajectories of length $\tau$ from a representative point in $g_i$
\If{all sampled trajectories satisfy $\mathcal{C}_u$}
    \State $\mathcal{G}_u \gets \mathcal{G}_u \cup \{g_i\}$
\ElsIf{at least one sampled trajectory satisfies $\mathcal{C}_s$}
    \State $\mathcal{G}_s \gets \mathcal{G}_s \cup \{g_i\}$
\Else      
    \State $\mathcal{G} \gets \mathcal{G} \cup \mathrm{SplitCell}(g_i)$
\EndIf
\State \Return $\mathcal{G}, \mathcal{G}_{s}, \mathcal{G}_{u}$
\end{algorithmic}
\end{algorithm}

Notably, to verify stages 2 and 3, each cell is required to sample $n_s$ trajectories of length $\tau$ and check whether all satisfy an unsafe condition or at least one satisfies the safe condition. Finally, in cases where neither safe nor unsafe conditions can be verified, one is required to either increase the resolution via the SplitCell routine (Algorithm~\ref{alg:split_a_cell}) or eventually declare the cell to be unsafe when the resolution is met.

\begin{algorithm}[ht]
\caption{SplitCell($g$)}
\label{alg:split_a_cell}
\begin{algorithmic}[1]
\State \textbf{Input:} Grid cell $g = \mathcal{B}_r(x) \in \mathcal{G}$
\State Let $(x_1, x_2, \dots, x_n) = x$
\State $P := \left\{ x + \frac{2r}{3} \cdot \boldsymbol{\delta} | \boldsymbol{\delta} \in \{-1, 0, 1\}^n \right\}$
\State \Return $\mathcal{G}_{split} := \{ \mathcal{B}_{\frac{r}{3}}(p)|p \in P\}$
\end{algorithmic}
\end{algorithm}
\section{Numerical Simulations}
\label{sec:simulations}

In this section, we validate the performance and the safety of our algorithm using a 3D evasion problem:

\begin{align*}
\dot{x} = \frac{d}{dt}
\begin{bmatrix}
x_1 \\ x_2 \\ x_3
\end{bmatrix}
=
\begin{bmatrix}
-v + v \cos x_3 + ux_2 \\
v \sin x_3 - ux_1 \\
- u
\end{bmatrix},
\end{align*}
with $[x_1, x_2]^T \in \mathbb{R}^2$ representing the relative planar location and $x_3 \in [0, 2\pi]$ the relative direction. $v \geq 0$ is the aircraft velocity and $u \in [-1,1]$ is the evader's angular velocity. A collision occurs if $\sqrt{x_1^2 + x_2^2} \leq 1$,
which defines a cylindrical collision set of radius 1 along the \(x_3\)-axis. Our goal is to determine the set of initial states that inevitably lead to a collision, regardless of the evader's actions.

\subsection{Results Comparison}
We use $\tau = 1$ s and a total number of control samples per cell $n_s = 3000$.  $V_{BRT}$ denotes the unsafe region volume computed via HJ reachability at $r_{\min} = 0.041$ (the grid resolution). We calculate the intersection ratio $(V_{BRT\cap h'\leq0} /V_{BRT})$ between our method ($\beta = \alpha = 0.05$) and HJ solutions at different precisions. As shown in Table~\ref{tab:hj_rea_diff_r} and Figure~\ref{fig:hj_diff_r}, low-precision HJ analysis underestimates unsafe regions, risking safety misjudgment. Our method guarantees complete containment of true unsafe regions at all precisions.

\begin{table}[ht]
    \caption{Comparison of the Fraction of the Unsafe Zone Volume Captured by Different Methods at  Different Precision}
    \centering
    \begin{tabular}{|c|c|c|c|c|}
        \hline
        Methods $\backslash$ $r_{\min}$ & 1.111 & 0.370 & 0.123 & 
        0.041\\
        \hline
        HJ BRT~\cite{hj_reachability} & 0.510 (\ding{55}) & 0.900 (\ding{55}) & 0.953 (\ding{55}) & 1 (\ding{51})\\
        \hline
        Recurrent Set & 1 (\ding{51}) & 1 (\ding{51}) & 1 (\ding{51}) & 1 (\ding{51})\\
        \hline
    \end{tabular}
    \label{tab:hj_rea_diff_r}
\end{table}

Beyond safety guarantees, Table~\ref{tab:hj_reachability_diff_precision} demonstrates our method's faster computation time at high precision through parallelization. 

\begin{table}[ht]
\caption{Comparison of Computation Time between HJ reachability and Recurrent Set Approximation under Different Precision, $\tau = 1$ s, $\alpha = 1$, $n_{\mathrm{s}} = 3000$}
    \centering
    \begin{tabular}{|c|c|c|c|c|}
        \hline
        Method $\backslash$ $r_{\min}$  & 1.111 & 0.370 & 0.123 & 0.041 \\
        \hline
        HJ BRT~\cite{hj_reachability} & 0.02 s (\ding{55}) & 0.19 s (\ding{55}) & 2.33 s (\ding{55}) & 83.73 s (\ding{51})\\
        \hline
        Recurrent Set & \textbf{0.13} s (\ding{51}) & \textbf{0.61} s (\ding{51})& \textbf{3.19} s (\ding{51})& \textbf{19.75} s (\ding{51})\\
        \hline
    \end{tabular}
    \label{tab:hj_reachability_diff_precision}
\end{table}
\vspace{-1.5em}
\begin{figure}[htbp]
  \centering
  \hspace{-3em}
  \subfigure[HJ Reachability]{
    \includegraphics[width=0.5\columnwidth]{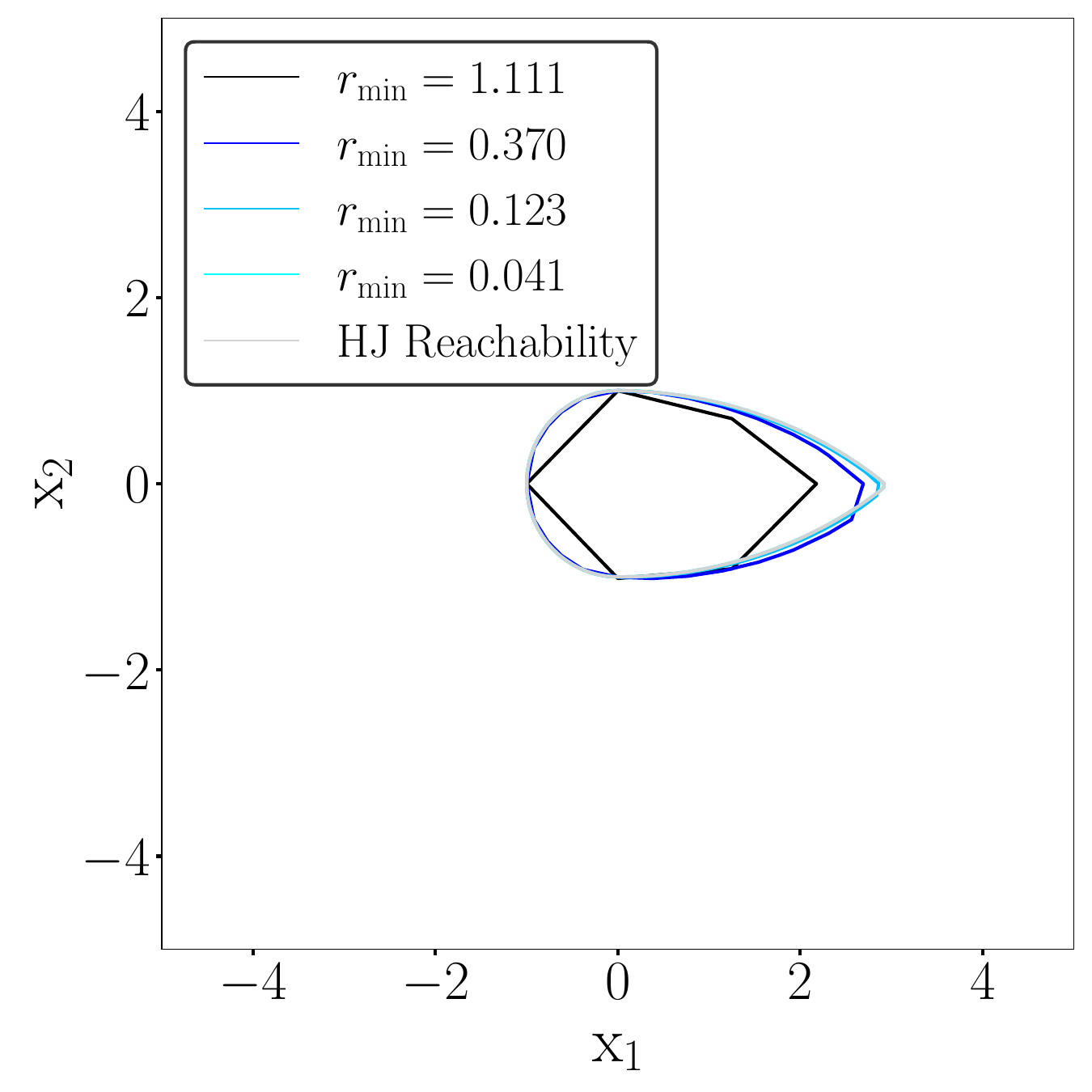}
    \label{fig:subfig1}}
  \hspace{-1.2em}
  \subfigure[Recurrent Set Approximation]{
    \includegraphics[width=0.56\columnwidth]{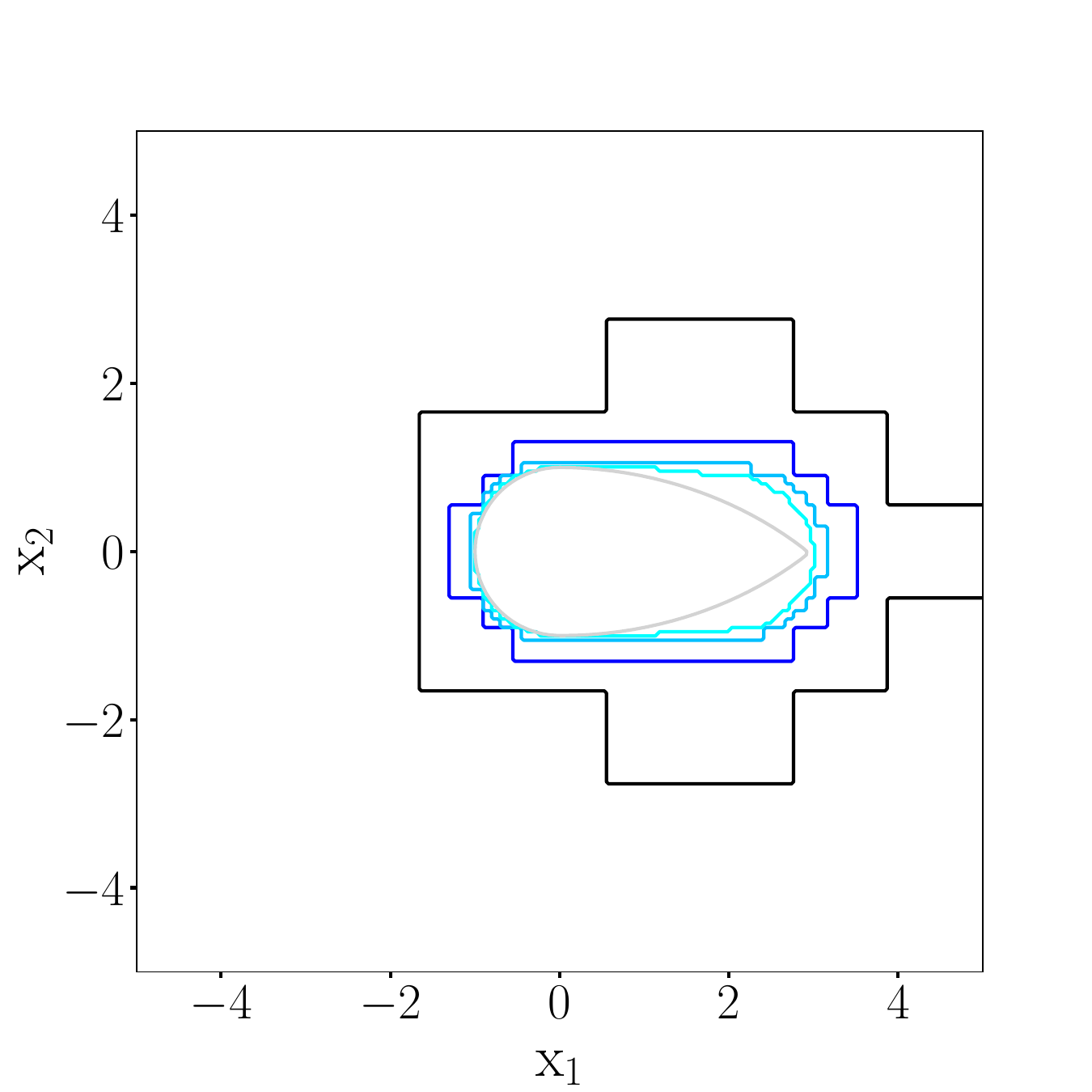}
    \label{fig:subfig2}}
  \hspace{-4.5em}
  \caption{Contour Plot of the Boundary of the Unsafe Region with Different Precision and Different Methods when $x_3 = \pi$}
  \label{fig:hj_diff_r}
\end{figure}

\subsection{Ablation Study}

Definition~\ref{def:recurrence} shows the recurrent set converges to the invariant set as $\tau\!\to\!0$. To gauge parameter effects on RCBF performance, we ran a sweep over two metrics: the normalized volume gap $(V_{\tau}-V_{\mathrm{BRT}})/V_{\mathrm{BRT}}$ and computation time $t$. Figure~\ref{fig:RCBF_charac} indicates an inverse $\tau$–accuracy trade-off: smaller $\tau$ reduces the volume gap but drives computation time up (roughly exponentially). Nonetheless, runtimes remain practical and safety is preserved for all tested parameters.

\begin{figure}[htbp]
  \centering
  \hspace{-2.8em}
  \subfigure[Volume Difference]{
    \includegraphics[width=0.5\columnwidth,height=0.16\textheight]{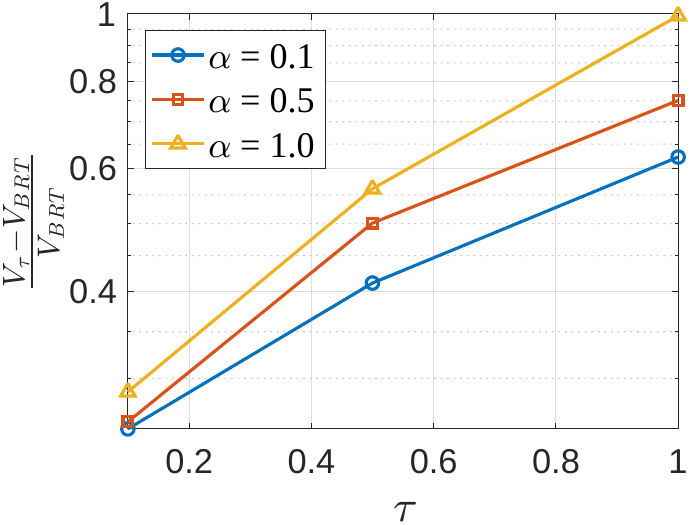}
    \label{fig:subfig1_charac}
  }
  \hspace{-1.1em}
  \subfigure[Computation Time]{
    \includegraphics[width=0.5\columnwidth,height=0.157\textheight]{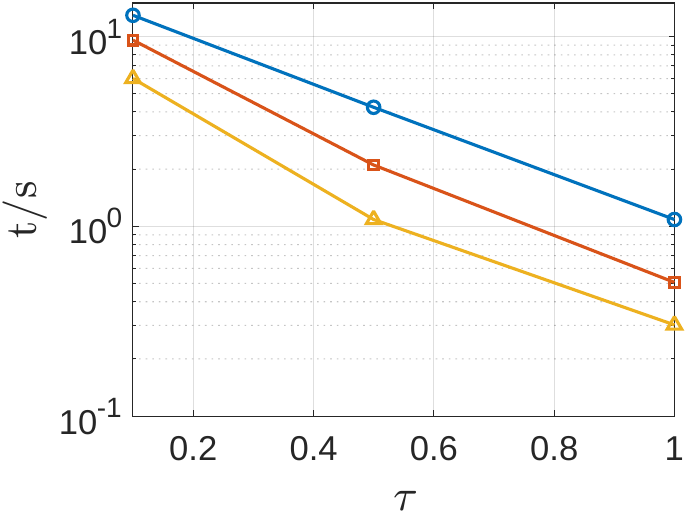}
    \label{fig:subfig2_charac}
  }
  \hspace{-2.5em}
  \caption{Volume gap and computation time versus $\tau$ (and $\alpha$); $n_{\mathrm{s}}{=}3000$, $r_{\min}{=}0.370$.}
  \label{fig:RCBF_charac}
\end{figure}

\vspace{-1em}
\section{Conclusion and Discussion}
\label{sec:conclusion}

We introduced Recurrent Control Barrier Functions (RCBFs), generalizing CBFs by enforcing finite-time ($\tau$) return rather than strict invariance. We proved that the signed distance to a $\tau$-recurrent set is a valid RCBF, yielding rigorous safety guarantees. A sampling-based algorithm approximates the safe region. Simulations demonstrate provably safe, though over-approximated, sets with competitive computational performance.

An approximation gap persists between computed and true safe sets. While denser sampling improves accuracy, the precise link between sampling parameters and error remains open. Future work will quantify this relationship and develop corresponding models to guide adaptive sampling for tighter guarantees.

\printbibliography

\ifthenelse{\boolean{arxiv}}{}{
\appendix
\subsection{Proof of Lemma~\ref{lemma:triangle_inequality}}
\label{lemma:tran_ineq}
\begin{proof}
    According to the assumption of system~\eqref{eq:Control_System}, since the system~\eqref{eq:Control_System} is uniformly continuous in $u$, thus $L_{u} = \max\limits_{t\in (0, \tau]} \|g(\phi(t,x,u))\| := \max\limits_{\|v\| = 1, t\in (0,\tau]}\|g(\phi(t,x,u)) v\| > 0$ exists. And the system~\eqref{eq:Control_System} is uniformly continuous in $u$, and Lipschitz continuous in $x$ for fixed control $u$, with a little abuse of the notation, for all the states $x'$ and $y'$ in trajectories $\phi(t,x,u)$ and $\phi(t,y,u)$, $\forall t\in (0,\tau]$ we have:
\begin{align*}
    & \|F(x, u) - F(y, u)\| \leq L \|x - y\|\\
    & \|F(x, u) - F(x, v)\| = \|g(x)(u-v)\| \leq L_{u}\|u-v\|
\end{align*}
Thus, according to Corollary 3.17 and Gr{\"o}nwall Comparison Lemma in~\cite{fb2024book}, we have: 
\begin{align*}
        & \|\phi(t,x,u) - \phi(t,y,u)\| \notag\\
        & \leq e^{Lt} \|x - y\| + L_{u} \int_{0}^{t}e^{L(t-s)}\|u(s) - u(s)\|ds\\
        & = e^{Lt}\|x-y\| \\
        & \leq re^{Lt},
\end{align*}
where equality is held because two trajectories have the same input trajectory. 

Suppose $x^{*}=\arg\min_{x^* \in \partial S} \mathrm{sd}(\phi(t,x,u), \mathcal{S}), y^{*}=\arg\min_{y^* \in \partial S} \mathrm{sd}(\phi(t,y,u), \mathcal{S})$. The three cases are analyzed as follows:

\textbf{Case 1}: $\phi(t,x,u)$ and $\phi(t,y,u)$ are both in $S$, and then we have
\begin{align*}
 & |\mathrm{sd}(\phi(t,x,u),S) - \mathrm{sd}(\phi(t,y,u), S)| \\
 = & |\|\phi(t,x,u) - x^{*}\| - \|\phi(t,y,u)-y^*\||\\
\le & |\|\phi(t,x,u) - x^{*}\| - \|\phi(t,y,u)-x^*\||\\
\le & \|\phi(t,x,u) - \phi(t,y,u)\| \\
\le & r e^{Lt},
\end{align*}
where the first equality follows from the definition, and the first inequality follows from the triangle inequality.

\textbf{Case 2}: $\phi(t,x,u)$ and $\phi(t,y,u)$ are both not in $S$, and then with the same reason, similarly, we have
\begin{align*}
 & |\mathrm{sd}(\phi(t,x,u),S) - \mathrm{sd}(\phi(t,y,u), S)|\\
 = & |\|\phi(t,x,u) - x^{*}\| - \|\phi(t,y,u)-y^*\||\\
\le & |\|\phi(t,x,u) - x^{*}\| - \|\phi(t,y,u)-x^*\||\\
\le & \|\phi(t,x,u) - \phi(t,y,u)\| \\
\le & re^{Lt},
\end{align*}

\textbf{Case 3}: One of $\phi(t,x,u)$ and $\phi(t,y,u)$ is in $S$ and the other not in $S$. Without loss of generality, we can assume that $\phi(t,x,u)$ is in $S$ and $\phi(t,y,u)$ is not in $S$. Then there at least exists a $\lambda \in [0, 1]$ such that $\lambda \phi(t,x,u) + (1-\lambda)\phi(t,y,u) \in \partial S$ and we denote $p^{*} := \lambda\phi(t,x,u) + (1-\lambda)\phi(t,y,u) \in \partial S$, thus, we have
\begin{align*}
 & |\mathrm{sd}(\phi(t,x,u),S) - \mathrm{sd}(\phi(t,y,u), S)|\\
 = & \|\phi(t,y,u)-y^*\| + \|\phi(t,x,u) - x^{*}\|\\
\leq &  \|\phi(t,x,u) - p^{*}\| + \|\phi(t,y,u)-p^*\||\\
= & \|\phi(t,x,u) - \phi(t,y,u)\|\\
\leq & re^{Lt},
\end{align*}
where the first equality and the first inequality follow from the definition of the signed distance function, and the second equality follows from the definition of $p^{*}$. In all cases, we obtain $|\mathrm{sd}(\phi(t,x,u),S) - \mathrm{sd}(\phi(t,y,u), S)| \leq re^{Lt}$ as required.
\end{proof}
}
\end{document}